\documentclass[fontsize=11pt,DIV=12]{scrartcl}
\usepackage[utf8]{inputenc}
\usepackage[english]{babel}
\usepackage[sort&compress,square,numbers]{natbib}

\usepackage{etoolbox}
\usepackage{graphicx}

\usepackage{amsmath}
\usepackage{amssymb}
\usepackage{bm}
\usepackage{dsfont}

\usepackage{varioref} %
\usepackage{hyperref}
\usepackage{amsthm} %
\usepackage{thmtools,thm-restate} %
\usepackage[capitalise]{cleveref} %

\usepackage{algorithm}
\usepackage[noend]{algpseudocode}

\usepackage{enumitem}
\usepackage{placeins}
\usepackage{xspace}
\usepackage{csquotes}

\usepackage[draft,footnote,nomargin,marginclue]{fixme}

\newcommand{\ie}{%
  that is\@\xspace%
}
\newcommand{\cf}{%
  cf.\@\xspace%
}

\newcommand{\Wlog}{%
  Without loss of generality\xspace%
}

\newcommand{\Accept}{\textbf{accept}}  
\newcommand{\Reject}{\textbf{reject}}
\newcommand{\DistVert}[1]{\textit{vertex #1}}
\newcommand{\DistAllVert}[1][]{\textit{every vertex #1}}
\newcommand{\DistNumRounds}[1]{\underline{\smash{#1 rounds}}}

\newcommand{\DistNumRoundsAllVert}[2]{\DistNumRounds{#1}, \textbf{let} \DistAllVert[#2]}

\algblockdefx{CompVert}{EndCompVert}[1]{\textbf{let} \DistVert{#1} \textbf{do}}{\textbf{\algorithmicend\ let}}
\makeatletter
\ifthenelse{\equal{\ALG@noend}{t}}%
  {\algtext*{EndCompVert}}
  {}%
\makeatother
\algblockdefx{CompAllVert}{EndCompAllVert}[1]{\textbf{let} \DistAllVert[#1] \textbf{do}}{\textbf{\algorithmicend\ let}}
\makeatletter
\ifthenelse{\equal{\ALG@noend}{t}}%
  {\algtext*{EndCompAllVert}}
  {}%
\makeatother

\newcommand{\dalg}[1]{%
  \mathcal{#1}
}
\newcommand{\dalgon}[2]{%
  \dalg{#1}_{#2}
}

\newcommand{\fstop}{%
  \, .
}

\newcommand{\innerprod}[2]{%
  \ensuremath{\langle #1, #2 \rangle}%
}
\newcommand{\hinnerprod}[2]{%
  \ensuremath{\left\langle #1, #2 \right\rangle}%
}

\newcommand{\vecto}[2][]{%
  \ensuremath{\vec{\boldsymbol{#2}}\ifstrempty{#1}{}{(#1)}}%
}

\newcommand{\defeq}{%
  \ensuremath{:=}
}

\newcommand{\minus}{%
  \ensuremath{\, \backslash \,}
}
\newcommand{\cupdot}{%
  \ensuremath{\mathbin{\dot{\cup}}}%
}

\newcommand{\onevector}[1]{%
  \mathds{1}_{#1}%
}
\newcommand{\Tr}{T}

\newcommand{\normp}[2]{%
  \ensuremath{\lVert #2 \rVert_#1}%
}
\newcommand{\hnormp}[2]{%
  \ensuremath{\left\lVert #2 \right\rVert_#1}%
}

\newcommand{\normb}[1]{%
  \ensuremath{\normp{2}{#1}}%
}
\newcommand{\hnormb}[1]{%
  \ensuremath{\hnormp{2}{#1}}%
}

\newcommand{\setr}{%
  \ensuremath{\mathbb{R}}%
}

\newcommand{\setn}{%
  \ensuremath{\mathbb{N}}%
}

\newcommand{\onerange}[1]{%
  \ensuremath{[#1]}%
}

\newcommand{\poly}{%
  \mathrm{poly}
}

\newcommand{\sep}{%
  \; | \;%
}

\newcommand{\OO}{%
  \ensuremath{\mathcal{O}}%
}
\newcommand{\OOt}{%
  \ensuremath{\tilde{\mathcal{O}}}%
}

\newcommand{\kdisc}{%
  \mbox{$k$-disc}\xspace%
}

\newcommand{\kdiscs}{%
  \mbox{$k$-discs}\xspace%
}

\newcommand{\fdisc}[2]{%
  \fkdisc{k}{#1}{#2}
}
\newcommand{\fkdisc}[3]{%
  \ensuremath{\mathrm{disc}_{#1}(#2, #3)}
}
\newcommand{\fdisci}[2]{%
  \fkdisci{k}{#1}{#2}
}
\newcommand{\fkdisci}[3]{%
  \ensuremath{\mathrm{disc}^*_{#1}(#2, #3)}
}

\newcommand{\neighbors}[1]{\Gamma(#1)}
\newcommand{\vol}[1]{\mathrm{vol}(#1)}
\newcommand{\girth}[1]{\mathrm{girth}(#1)}
\newcommand{\dist}[2]{\mathrm{dist}(#1, #2)}

\newcommand{\edgecut}[3]{%
  \ensuremath{#1(#2, #3)}
}

\newcommand{\Adj}{A}
\newcommand{\Deg}{D}
\newcommand{\Id}{I}
\newcommand{\Nlap}{L}

\newcommand{\Walk}[1][]{W\ifstrempty{#1}{}{(#1)}}
\newcommand{\Walkex}[2][]{W^{#2}\ifstrempty{#1}{}{(#1)}}

\newcommand{\Nwalk}[1][]{N\ifstrempty{#1}{}{(#1)}}

\newcommand{\Nwalkew}{\mu}
\newcommand{\Nwalkev}{\vecto{f}}

\newcommand{\statdist}[1][]{\vecto{\pi}_{#1}}
\newcommand{\Degr}{D}
\newcommand{\degr}[2][]{d_{\ifstrempty{#1}{}{#1}}(#2)}

\newcommand{\setcond}[1]{\mathrm{cond}(#1)}
\newcommand{\graphcond}[1]{\Phi(#1)}

\declaretheoremstyle[
  notefont=\normalfont\bfseries,
  bodyfont=\normalfont\itshape
  ]{normal_style}
\declaretheoremstyle[
  notefont=\normalfont\bfseries,
  qed=$\scriptstyle\blacksquare$
  ]{definition_style}
\declaretheoremstyle[
	headfont=\normalfont\itshape,
  notefont=\normalfont\itshape,
  notebraces={}{},
  qed=\qedsymbol
  ]{proof_style}

\declaretheorem[
  name=Theorem,
  numberwithin=section,
  style=normal_style]{theorem}

\declaretheorem[
  name=Lemma,
  numberlike=theorem,
  style=normal_style]{lemma}
\declaretheorem[
  name=Proposition,
  numberlike=theorem,
  style=normal_style]{proposition}

\declaretheorem[
  name=Definition,
  numberlike=theorem,
  style=definition_style]{definition}

\newcommand{\dcm}[3]{\mathrm{DCM}(#1, #2, #3)}
\newcommand{\dna}[2]{\mathrm{DNDA}(#1, #2)}

\makeatletter
\patchcmd{\ALG@doentity}{\item[]\nointerlineskip}{}{}{}
\makeatother

\newcommand{\proofref}[1]{\marginline{$\rightarrow$~p.~\pageref{#1}}}

\newlist{inenum}{enumerate*}{1}
\setlist[inenum]{label=(\roman*)}

\title{Distributed Testing of Conductance\footnote{The research leading to these results has received funding from the European Research Council under the European Union's Seventh Framework Programme (FP7/2007-2013)/ ERC grant agreement n$^\circ$~307696.}}
\author{Hendrik Fichtenberger\thanks{TU Dortmund, Dortmund (\texttt{hendrik.fichtenberger@tu-dortmund.de})} \and Yadu Vasudev\thanks{Department of Computer Science and Engineering, IIT Madras, Chennai (\texttt{yadu@cse.iitm.ac.in})}} 
\date{}%

\begin{document}
  \maketitle

  \begin{abstract}
    We study the problem of testing conductance in the setting of distributed computing and give a two-sided tester that takes $\OO(\log(n) / (\epsilon \Phi^2))$ rounds to decide if a graph has conductance at least $\Phi$ or is $\epsilon$-far from having conductance at least $\Phi^2 / 1000$ in the distributed CONGEST model. We also show that $\Omega(\log n)$ rounds are necessary for testing conductance even in the LOCAL model. In the case of a connected graph, we show that we can perform the test even when the number of vertices in the graph is not known a priori. This is the first two-sided tester in the distributed model we are aware of. A key observation is that one can perform a polynomial number of random walks from a small set of vertices if it is sufficient to track only some small statistics of the walks. This greatly reduces the congestion on the edges compared to tracking each walk individually.
  \end{abstract}

  \section{Introduction}

  Graphs arise as a natural model of many large data sets in applications like examining social networks, and analyzing structural properties of graphs is a fundamental computational problem. However, using exact algorithms to solve this task is often not an option because even linear time or space complexity exceeds available resources. Algorithms that compute an approximate result are more useful in these scenarios. Various frameworks have been studied to analyze the limits and merits of such algorithms in theory.

  Property testing algorithms derive approximate decisions by probing small parts of the input only. A tester for a graph property $\mathcal{P}$ is a randomized algorithm that, with high constant probability, accepts inputs that have $\mathcal{P}$ and rejects inputs that are $\epsilon$-far from having the property $\mathcal{P}$, \ie, at least an $\epsilon$-fraction of the edges has to be modified to make the graph have the property $\mathcal{P}$.  Testing graph properties in the classic, sequential computing model has been studied quite extensively.  Two-sided error testers may err on all graphs, while one-sided error testers have to present a witness when rejecting a graph. See \citep{GolInt,GolInt10,GolPro98} for introductions and surveys.

  Property testing in the distributed CONGEST model was first studied by \citet{BrakerskiP11} and later more thoroughly by \citet{Censor-HillelFSV16}. In this model, each vertex of the graph is equipped with a processor that has a unique identifier of size $\OO(\log n)$ and it knows only its neighboring vertices. The vertices of the graph communicate with each other in synchronized rounds such that in each round only communication of length $\OO(\log n)$ is allowed on every edge. Finally, every vertex casts a vote and a decision rule is applied on all votes to derive the answer of the tester. The complexity measure is the amount of rounds required to test the property. Edge congestion and round complexity strictly limit the amount of information on the whole graph that a single vertex can gather.

  In \citep{Censor-HillelFSV16}, it is shown that many one-sided error testers for dense graphs carry over from the sequential to the distributed setting. Furthermore, tight logarithmic bounds for testing bipartiteness and cycle-freeness in bounded degree graphs are proved. In \citep{FraigniaudRST16,EveThr17}, subgraph-freeness is studied for subgraphs on at most five vertices, trees and cliques.
  
  \subsection{Our Results}

  In this paper we study the problem of testing conductance of undirected graphs in the distributed model. We present a two-sided error \emph{distributed testing algorithm} in the CONGEST model for testing conductance, which is also the first two-sided error distributed tester we are aware of.

  \begin{theorem}
    \label{thm:upper-bound}
    Testing whether a graph $G=(V,E)$ has conductance at least $\Phi$ or is $\epsilon$-far from having conductance at least $\Phi^2 / 1000$ with two-sided error has complexity $\OO(\log (|V|+|E|) / (\epsilon \Phi^2))$ in the CONGEST model.
  \end{theorem}
  
  Our algorithm is based on the idea of the classic tester in \citep{KalExp11} for bounded degree graphs, \ie, random walks mix rapidly in graphs with high conductance and they mix slowly for at least a small fraction of start vertices in graphs that are $\epsilon$-far from having high conductance. In bounded degree graphs, the discrepancy of the distribution of a random walk and its stationary distribution is measured by the walk's collision probability. However, in general graphs the endpoint distributions weighted by the vertices' degrees play a major role. This becomes an obstacle because one has to ensure that vertices of high degree are not missed.

  In the distributed model, one has to take care of edge congestion too. Simulating $\omega(1)$ random walks while keeping them distinguishable is very costly. However, one key observation here is that for approximating the discrepancy, it is sufficient to maintain only some statistics of the random walks, which reduces the congestion significantly. In particular, one has to transfer only the number of random walks that pass through an edge for each of a constant number of start vertices. We exploit this and some other properties of graphs with high conductance.
  
  If the graph is connected, we prove that the size of the input graph is not required to be known a priori to perform the test. On the other hand, there exists no tester for disconnected graphs if no prior knowledge of the graph is assumed at all. Since communication between two connected components is not allowed, we cannot distinguish a graph~$G$ with high conductance from a graph~$G'$ that is composed of two isolated copies of~$G$.

  We complement this result by showing that any distributed tester with this gap requires $\Omega(\log (n+m))$ rounds of communication in the (stronger) LOCAL model.

  \begin{theorem}
    \label{thm:lower_bound_nn}
    Testing whether a graph $G=(V,E)$ has conductance at least $\Phi$ or is $\epsilon$-far from having conductance at least $c \Phi^2$ requires $\Omega(\log (|V|+|E|))$ rounds of communication in the LOCAL model (for constants $c, \epsilon, \Phi$).
  \end{theorem}  

  For the lower bound, we construct two distributions on graphs of high and low conductance respectively such that the vertices' neighborhoods of radius $\Omega(\log n)$ are isomorphic for both. The idea is that within only $\OO(\log n)$ rounds, all vertices receive the same information up to isomorphism and therefore cannot distinguish between the two distributions. It seems possible that the distributed algorithm can glean information about the two distributions from the vertex labels of the subgraphs it has seen. For example, certain sets of labeled subgraphs might be present in (many of) the graphs with high conductance that are absent in (many of) the graphs with low conductance. However, we show that this would not give sufficient information. Since the local views of the vertices have a large overlap, our technique to rule out this issue differs from the collision-based argument that is often used in the classic setting of property testing.

  \subsection{Related Work}
  
  In the classic (\ie, non-distributed) setting of property testing, the problem of testing conductance in bounded degree graphs was first studied in \citep{GolTes00,CzumajS07}. \Citet{KalExp11} and \citet{NacTes10} give $\OOt(\Phi^{-2}\sqrt{n})$-query testers that accept graphs that have conductance at least~$\Phi$ and reject graphs that are $\epsilon$-far from having conductance at least $\Omega(\Phi^2)$. An algorithm for testing the cluster structure of graphs has been proposed in \citep{CzuTes15}. Testing conductance in unbounded degree graphs in the stronger rotation map model with query complexity roughly $\OOt(\Phi^{-2} \sqrt{m})$ was studied in~\citep{LiTes11}. Testing conductance properties restricted to small sets has been studied in \citep{LiTes15}. The optimal query complexity for testing conductance in general graphs of unbounded degree is still open.

  We review some results that use related techniques and discuss similarities and differences compared to our approach. In the CONGEST model, random walks have been analyzed by \citet{Censor-HillelFSV16} to design a tester for bipartiteness. The idea there is to perform a constant number of random walks from every vertex and to test if two such walks intersect in a cycle of odd length. Therefore, the algorithm needs to keep the exact trace of each of the random walks. In contrast, we are only interested in the start vertex and the (current) end vertex of a random walk. As a result we can perform polynomially many random walks from a constant number of vertices in the graph.

  Distributed random walks have been studied in \cite{DasDis13} and \cite{MolDis17} for computing the mixing time for random walks starting from a fixed vertex. In particular, \cite{MolDis17} show that one can approximate the mixing time~$\tau_v$ of a vertex $v$ in $\OO(\tau_v \log n)$ rounds by running $\poly(n)$ random walks~$v$ and comparing their endpoint distribution to the stationary distribution. The graph's mixing time $\tau = \max_v \tau_v$ relates to the conductance by $c_1\Phi^2 / \log n) \leq 1/ \tau \leq c_2\Phi$. A straightforward approach based on \citep{MolDis17} leads to an $\OO(n \log^2 (n) / \Phi^2)$ round algorithm for \emph{approximating} $\Phi$ with a multiplicative gap of $\Theta(\Phi / \log n)$. In comparison, our \emph{tester's} gap does not depend on $n$ and its complexity is only logarithmic in $n$. One reason is that if the graph is \emph{far from} having conductance $\Omega(\Phi^2)$, there exist many vertices with large mixing times compared to the case that the graph has conductance $\Phi$ (see the proof of \cref{thm:upper-bound} for details). This is not necessarily the case if the graph is not $\epsilon$-far from having conductance $\Omega(\Phi^2)$.

  \section{Preliminaries}
  
  Let $G=(V,E)$ be a graph and let $S, T \subseteq V, S \cap T = \emptyset$ be sets of vertices. For simplicity, we denote $|V|$ and $|E|$ by $n$ and $m$ respectively for the graph $G$ at hand. Let $\degr{v}$ be the degree of vertex $v \in V$. We write $\bar{S}$ for the set $V \backslash S$. The set of vertices in $\bar{S}$ that are adjacent to some $u \in S$ is denoted by $\neighbors{S}$. The \emph{volume} of $S$ is the sum of degrees of vertices in $S$, \ie, $\vol{S} \defeq \sum_{v \in S} \degr{v}$. The cut between $S$ and $T$ is denoted by $\edgecut{E}{S}{T} = E \cap (S \times T)$. For a set $S\subseteq V$ such that $\vol{S} \leq \vol{\bar{S}}$, the conductance of $S$ is $\setcond{S} = |\edgecut{E}{S}{\bar{S}}|/\vol{S}$. The conductance of $G$ is defined as $\graphcond{G} = \min_{S\subset V} \setcond{S}$. %

  \subsection{Distributed Computing}
\label{sec:distr-comp}
  In the distributed computational model, a computation network $G=(V,E)$ with a processor associated to each vertex $v \in V$ is given. Each processor $v$ has access to numbered communication channels to its neighbors in $G$. Additionally, it may have some specific input $I(v)$. The computation operates in synchronized rounds that are divided into three phases. In each round, each processor may do some local computation first, then it may send a message to each of its neighbors, and finally it receives the messages sent from its neighbors.

  \begin{definition}[Distributed Computational Model, DCM]
    \label{def:dcm}
    Let $G=(V,E)$ be a graph and $p_G = (p_v)_{v \in V}$ with $p_v : \onerange{\degr{v}}
    \rightarrow \neighbors{v}$ be a bijective function, \ie, an adjacency list representation of
    $G$. Let $I : V \rightarrow \{ 0, 1 \}^*$ be a mapping from the set of vertices to bit
    strings.  An instance of the distributed computational model on $G$, $p_G$ and $I$,
    $\dcm{G}{p_G}{I}$, is defined as follows.  Each vertex $v \in V$ is a processor that has
    communication access to its neighbors $p_v(1), \ldots, p_v(\degr{v})$ by ports numbered $1,
    \ldots, \degr{v}$.  The model operates in synchronized rounds, where each round $r$ consists of
    three phases: 
    \begin{inenum}
      \item Each vertex performs local computation,
      \item each vertex $v$ sends a message to its neighbor $p_v(i)$, 
        denoted	$s_r(v, i)$, for all $i \in \degr{v}$,
      \item each vertex $u$ receives a message from its neighbor $p_u(j)$, for all $j \in \degr{u}$.
    \end{inenum}
    The distributed computational model DCM is the set of all instances $\dcm{G}{p_G}{I}$.
  \end{definition}

  The LOCAL model is the subset of the DCM such that for each vertex $v \in V$, the input $I(v)$ is only $n$ and a numerical vertex identifier from $[n^c]$ for some universal constant $c$. The CONGEST model is the subset of the LOCAL model such that the size of each message $s_r(v,i)$ is restricted to $c \log n$ bits.

  A distributed network decision algorithm $\dna{\dalg{A}}{O}$ is an algorithm $\dalg{A}$ that is deployed to the vertices of a DCM to decide a property of an instance of the model.  In particular, the output of $\dalg{A}$ is a single bit, and the final decision is obtained by applying a function $O(\cdot)$ to the union of all vertices' answers.

  \begin{definition}[Distributed Network Decision Algorithm]
  	\label{def:dna}
    Let $\dalg{A}$ be an algorithm that takes a bit string as input and outputs a single bit, and
    let $O : \{ 0, 1 \}^* \rightarrow \{ 0, 1 \}$ be a function.  When the distributed network
    decision algorithm $\dna{\dalg{A}}{O}$ is run on an instance $\dcm{G}{p_G}{I}$, a copy of
    $\dalg{A}$ is deployed to every vertex $v$ with input $I(v)$ and run in parallel as described in \cref{def:dcm}.  We refer to the copy of $\dalg{A}$ deployed to $v$ by $\dalgon{A}{v}$.  When
    every vertex $v_i$ has terminated its computation with output bit $b_{v_i}$, the decision of
    $\dna{\dalg{A}}{O}$ is $O(b_{v_1} b_{v_2} \cdots b_{v_n})$.
  \end{definition}

\subsection{Distributed property testing}
\label{sec:distr-prop-test}

A distributed property testing algorithm is a distributed algorithm as defined in \cref{def:dna} that accepts graphs that have a property, and rejects graphs that are $\epsilon$-far from the property. We say that a graph $G$ with $n$ vertices and $m$ edges is $\epsilon$-far from a property ${\cal P}$ if at least $\epsilon m$ edges of $G$ have to be modified to make the new graph have the property ${\cal P}$. 

A one-sided error distributed $\epsilon$-test accepts all graphs with property ${\cal P}$, whereas it rejects, with probability at least $2/3$, all graphs that are $\epsilon$-far from the property. In this paper, we give a two-sided (error) property tester that is also allowed to err, with probability at most $1/3$, when the graph has the property.

\begin{definition}[Two-sided tester]
  A two-sided (error) distributed $\epsilon$-test for a property ${\cal P}$ is a $\dna{\dalg{A}}{O}$, where $O(b_{v_1} b_{v_2} \cdots b_{v_n})=1$ iff $b_{v_i} = 1$ for all $v_i \in V$ such that the following conditions hold:
  \begin{itemize}
  \item If $G$ has the property ${\cal P}$, then, with probability at least $2/3$, $b_{v_i}=1$ for all $v_i \in V$.
  \item If $G$ is $\epsilon$-far from ${\cal P}$, then, with probability at least $2/3$, there exists a $v_i \in V$ such that $b_{v_i}=0$. \qedhere
  \end{itemize}
  \label{defn:dist-pt}
\end{definition}

The guarantees given by our tester are actually a bit stronger in the sense that the tester can be modified such that either $b_v = 0$ or $b_v = 1$ for all $v \in V$ simultaneously.
See \cref{sec:open_problems} for a discussion of the acceptance behavior.

  \section{Testing Using Random Walks}
  \label{sec:upper_bound}

  In this section we will present the distributed algorithm for testing whether a graph has conductance at least~$\Phi$ or is $\epsilon$-far from having conductance at least~$\Phi^2 / 1000$. The core idea of the algorithm is to perform random walks from a small set of vertices and test whether these walks converge to the stationary distribution rapidly, which is the case for graphs with high conductance. It is based on the ideas of \citet{KalExp11} and \citet{GolTes00}.

  Before we describe the algorithm, we give a few useful definitions and lemmas. A {\em lazy random walk} on a graph $G=(V,E)$ on $n$ vertices is a random walk on the graph, where at each vertex $v$ the walk chooses to stay at $v$ with probability $1/2$ and chooses a neighbor $u$ with probability $1/(2\degr{v})$. The walk matrix $\Walk = [w_{uv}]_{u,v \in \onerange{n}}$ is defined by $w_{uv} \defeq 1/2$ if $u=v$, $w_{uv} \defeq 1/(2 \degr{v})$ if $u \neq v, (u,v) \in E$ and $w_{uv} \defeq 0$ otherwise. Notice that for irregular graphs, $\Walk$ is not symmetric. To analyze these random walks, one can draw on the \emph{normalized walk matrix}, which is a symmetric matrix similar to $\Walk$. The normalized walk matrix~$\Nwalk$ of~$G$ is $\Degr^{-1/2} \Walk \Degr^{1/2}$, where $\Degr$ is the diagonal matrix with $\Degr(u,u) \defeq d(u)$.

  Since $\Nwalk$ is a real symmetric matrix, it has real eigenvalues. Let $1 = \Nwalkew_1, \ldots, \Nwalkew_n \geq 0$ be its eigenvalues, and let $\{ \Nwalkev_i \}_{i \in \onerange{n}}$ be its orthonormal eigenbasis. We have $\Nwalkev_1 = \sqrt{\statdist}$, where $\statdist$ is the random walk's \emph{stationary distribution}. In particular, it is well known that $\statdist[v] = \degr{v} / (2m)$. For more details on spectral graph theory, refer to \citep{ChuSpe}.

  It is well known that graphs with high conductance have small diameter.
  \begin{lemma}[{\citep[\cf Theorem~2]{ChuDia89}}]
    Let $G=(V,E)$ be a graph with conductance $\Phi$.
    The diameter of $G$ is at most $(3 / \Phi) \ln(m)$.
    \label{cor:diameter}
  \end{lemma}

  \Citet{SinAlg93} proved that there is a tight connection between the conductance and the mixing time of random walks. In particular, the $L_2$ distance of any starting distribution $\vecto{\pi'}$ to $\statdist$ after $\Phi^{-2} \log n$ steps is $\OO(1/n)$.

  \begin{lemma}[{\citep[\cf Theorem~2.5]{SinAlg93}}]
    \label{thm:high-cond-close}
    Let $G=(V,E)$ be a graph with conductance $\Phi$.
    For any starting distribution $\vecto{\pi'}$, it holds that
    $%
      \normb{\Walk^\ell \vecto{\pi'} - \statdist}
      \leq \left( 1- \Phi^2/2 \right)^\ell
    $.%
  \end{lemma}

  \subsection{Algorithm}

  We discuss the algorithm from a global point of view instead of describing an algorithm $\mathcal{A}$ for a single vertex to provide a better explanation of the interactions between vertices.
  
  \Cref{cor:diameter} implies that if the graph has high conductance, then it has diameter $\OO(\log n)$, which we want to use as an assumption in the algorithm later. To test the diameter, we perform a BFS of depth $\OO(\log n)$ of the graph starting from an arbitrary vertex. Initially, every vertex chooses itself as root of the BFS and announces itself as root to all its neighbors. To break the symmetry between the vertices, a vertex accepts every vertex with a lower identifier than its current root as new root and forwards its messages. If the diameter is $\OO(\log n)$, a unique root has been chosen after $\OO(\log n)$ rounds and every vertex knows its parent and its children in the BFS tree. Otherwise, at least one of the remaining candidates will reject. \Vref{alg:bfs} gives a formal description of the BFS.

  From now on, assume that the diameter is $\OO(\log n)$. Using the previously computed BFS tree, we can compute the number of edges in the graph by summing up vertex degrees from the leaves to the root and transmitting this number to all vertices afterwards. \Vref{alg:sum} describes the procedure in detail.

  The key technical lemma from \citep{KalExp11} for bounded degree graphs states that if a graph is $\epsilon$-far from having conductance $\Omega(\Phi^2)$, then there exists a $\Omega(\epsilon)$-fraction of \emph{weak} vertices such that random walks starting from these vertices converge only slowly to the stationary distribution. Therefore, a sample $S \subset V$ of size $\OO(1)$ will likely contain a weak vertex.\footnote{Technically, we sample each vertex $v$ independently into $S$ with probability $\Theta(\degr{v} / \epsilon m)$. By Markov's inequality, we may reject if $S$ is much larger than its expected size.} We extend this lemma to unbounded degree graphs. Then, we perform $N = \OO(n^{100})$ random walks of length $\ell = \OO(\log n)$ starting from each of the vertices in $S$ to approximate the rate of convergence.
  
  The crucial point here is that in each round of the algorithm, we do not send the full trace of every random walk. Instead, for every origin $v \in S$, every vertex $u \in V$ only transmits the total number of random walks that are leaving it through an edge $(u,w)$ to its neighbor $w \in \neighbors{u}$. Since the size of $S$ is constant, we require $\OO(\log n)$ bits per edge to communicate this. On the other hand, this information is sufficient because we are only interested in the distribution of endpoints of the lazy random walks for every $v \in S$. \Cref{alg:random-walk} gives a formal description of this procedure. Finally, the estimated distribution of endpoints is used to approximate the distance to the stationary distribution for each $v \in S$. The whole algorithm is summarized in \cref{alg:cond-tester}.

  \begin{algorithm}
    \caption{Conductance tester}
    \label{alg:cond-tester}
    \begin{algorithmic}[1]
      \Procedure{TestConductance}{$G = (V, E)$, $n$, $\Phi$}
        \State \textsc{BFS}($G$, $6/\Phi \ln n$) \Comment{construct BFS of depth $6/\Phi \log n$, \cref{alg:bfs}}
        \If{BFS visited less than $n$ vertices}
          \Reject
        \EndIf
        \State $m \gets$ \textsc{AggegrateSum}($G$, $12/\Phi \ln n$, $f$) \Comment{$f(v) \defeq \degr{v}/2$, \cref{alg:sum}}
        \CompAllVert{$v \in V$}
        	\State with probability $\min\{1,10^4\degr{v} / 2\epsilon m\}$, mark $v$ \label{lin:choose_s}
        \EndCompAllVert
        \State $S \gets$ marked $v$, $r \gets$ root of BFS tree
        \If{$|S| > 10^5 / \epsilon$}
        	\Reject \label{lin:large-set}
       	\EndIf
       	\State \textsc{RandomWalk}($G$, $S$, $40/\Phi^2 \cdot \log n$, $n^{100}$) \Comment{compute local $s_{v,u}$, \cref{alg:random-walk}}
        \ForAll{$v \in S$}
          $s_v \gets$ \textsc{AggegrateSum}($G$, $12/\Phi \ln n$, $f$) \Comment{$f(u) \defeq s_{v,u}$, Alg.~\ref{alg:sum}}\label{lin:aggregate_sv}
        \EndFor
        \CompAllVert{$v \in V$}
          \If{$s_v \leq m^{-15}$ for all $v \in S$}
            \Accept \label{line:conductance_accept}
          \Else
            ~\Reject \label{line:conductance_reject}
          \EndIf
        \EndCompAllVert
      \EndProcedure
    \end{algorithmic}

  \end{algorithm}

  \begin{algorithm}
    \caption{Perform random walks}
    \label{alg:random-walk}
    \begin{algorithmic}[1]
      \Procedure{RandomWalk}{$G$,$S$,$\ell$,$N$}
        \CompAllVert{$v \in S$}
           \State $P_v \gets \{u_1,\cdots,u_N\}$ where each $u_i$ is chosen indep. according to $\Walk \vecto{e_v}$
          \ForAll{vertex $w \in P_v$ chosen $n_w$ times}
            send $(v,n_w)$ to $w$
          \EndFor
        \EndCompAllVert
        \For{\DistNumRoundsAllVert{$\ell$}{$v$}}
          \If{$v$ receives $(v_1,n_1), (v_2, n_2), \cdots, (v_k, n_k)$}
            \ForAll{$(v_i, n_i)$}
              \State $P_{v_i} \gets \{u_1,\cdots,u_{n_i}\}$ where $u_j$ is picked indep. according to $\Walk \vecto{e_v}$
              \ForAll{vertex $w \in P_{v_i}$ chosen $n_w$ times}
                append $(v_i, n_w)$ to $L_w$
              \EndFor
            \EndFor
            \ForAll{$w \in \neighbors{v}$}
              send $L_w$ to $w$
            \EndFor
          \EndIf
        \EndFor
        \CompAllVert{$u \in V$}
          \If{$u$ receives $(v_1,n_1), (v_2, n_2), \cdots, (v_k, n_k)$}
            \ForAll{$v \in S$}
              \State $\widehat{\Walk}^\ell_{v,u} \gets \sum_{v_i = v} n_i / N$
              \If{$\widehat{\Walk}^\ell_{v,u} \leq 2 m^{-2}$}
                \Reject \label{lin:small_W_reject}
              \EndIf
              \State $s_{v,u} \gets (\widehat{\Walk}^\ell_{v,u})^2 - \widehat{\Walk}^\ell_{v,u} \frac{\degr{v}}{2m} + \frac{\degr{v}^2}{4m^2}$
            \EndFor
          \EndIf
        \EndCompAllVert
      \EndProcedure
    \end{algorithmic}
  \end{algorithm}

  First, we show that either the estimates $\widehat{\Walk}^\ell_{v,u}$ of \cref{alg:random-walk} are good or the algorithm rejects in line~\ref{lin:small_W_reject} because $G$ has low conductance. The proof is given in the appendix.
  \begin{restatable}{lemma}{thmapproxenddist}
    \label{thm:approx_end_dist}
    Consider \cref{alg:random-walk}.
    For every $v, u \in V$, it holds with probability at least $1-m^{-10}$ that
    \begin{inenum} 
      \item $|\widehat{\Walk}^\ell_{v,u} - \Walkex[v,u]{\ell}| \leq m^{-20}$ and, conditioned on the previous, \item if $\widehat{\Walk}^\ell_{v,u} < m^{-2}$ then $G$ has conductance less than~$\Phi$.
    \end{inenum}
  \end{restatable}
  \proofref{proof:approx_end_dist}
  
  Furthermore, \cref{thm:approx_end_dist} implies that the estimates $s_v$ in \cref{alg:cond-tester} (see line~\ref{lin:aggregate_sv}) are also good if \cref{alg:random-walk} has not rejected before.

  \begin{lemma}
    \label{thm:approx_summand}
    Consider \cref{alg:cond-tester}.
    With probability at least $1-m^{-8}$ it holds for every $v \in S$  in line~\ref{lin:aggregate_sv} that
    $%
      \left| \normb{\Walkex[v, \cdot]{\ell} - \statdist}^2 - s_v \right|
      \leq 3m^{-19}
    $.%
  \end{lemma}
  \begin{proof}
    Let $v \in S$. We have the following equality for the discrepancy of the distribution of the random walks' endpoints that start at $v$ and the stationary distribution:
    \begin{equation}
      \label{eq:walk_dist_equality}
      \normb{\Walkex[v, \cdot]{\ell} - \statdist}^2 = \sum_{u \in V}\left( (\Walkex[v,u]{\ell})^2 -\Walkex[v,u]{\ell} \frac{d(u)}{2m} + \frac{d(u)^2}{4m^2} \right) \fstop
    \end{equation}
    By \cref{thm:approx_end_dist}, we know that for every $u \in V$ we have
    $%
      |\widehat{\Walk}^\ell_{u,v} - \Walkex[v,u]{\ell}| \leq m^{-20}
    $%
    with probability $1 - 1/m^9$.
    Using $\Walkex[v,u]{\ell} \leq 1$, we have
    $\left| (\Walkex[v,u]{\ell})^2 - \Walkex[v,u]{\ell} \frac{\degr{u}}{2m} + \frac{\degr{u}^2}{4m^2} - s_{v,u} \right| \leq 3m^{-20}$.
    Combining this with \cref{eq:walk_dist_equality}, a union bound over all $u \in V$ implies that with probability at least $1 - n \cdot m^{-10} \geq 1 - m^{-9}$, we have that
    $\left| \normb{\Walkex[v,\cdot]{\ell} - \statdist}^2 - \sum_{u \in V} s_{v,u} \right| \leq 3m^{-19}$.
    A union bound over all $v \in S$ gives that with probability at least $1 - |S| / m^{-9} \geq 1 - m^{-8}$, $\left| \normb{\Walkex[v,\cdot]{\ell} - \statdist}^2 - \sum_{u \in V} s_{v,u} \right| \leq 3m^{-19}$.
  \end{proof}

  \subsection{Completeness and Soundness}

  The proof of completeness is a straightforward application of the results from the previous section.

  \begin{lemma}[Completeness]
    \label{thm:completeness}
    Let $G(V,E)$ be a graph with conductance at least~$\Phi$. Then, with probability at least $2/3$, each vertex in $G$ returns $\Accept$ when it runs \cref{alg:cond-tester}.
  \end{lemma}
  \begin{proof}
    The probability that the algorithm rejects in Line~\ref{lin:large-set} of \cref{alg:cond-tester} is at most $1/10$, and we assume, for the remainder of the proof, that this event did not occur.
    If $G$ has conductance at least $\Phi$, then from \cref{thm:high-cond-close} we know that $\normb{\Walk^\ell(\cdot,v) - \statdist}^2 \leq \left( 1 - \Phi^2/2 \right)^{2\ell} \leq \exp(- \Phi^2\ell/2 ) \leq m^{-20}$ for every vertex $v$.
    \Cref{thm:approx_summand} implies that with probability at least $9/10$, it holds that $\left| \normb{\Walkex[\cdot, v]{\ell} - \statdist}^2 -  s_v \right| \leq 3m^{-19}$.
    Conditioning on this event, every vertex accepts in line~\ref{line:conductance_accept} of \cref{alg:cond-tester}.
  \end{proof}

  To complete the analysis of the tester, we show that whenever the graph is $\epsilon$-far from having conductance $\Omega(\Phi^2)$, the tester rejects with probability at least $2/3$. To this end, we actually show that if the volume of weak vertices is small, then the graph can be converted to another graph $G'$ by modifying at most $\epsilon m$ edges such that the conductance is $\Omega(\Phi^2)$. The idea of the analysis is due to \citet{KalExp11}, who analyzed a classic property tester for testing expansion in graphs with vertex degrees bounded by a constant. We deviate from their analysis where it becomes necessary to take care of arbitrary vertex degrees.
  
  Let a vertex $v \in V$ be called \emph{weak} if $\normb{\Walk^\ell(v,\cdot) - \statdist} > 6m^{-15}$. The following lemma states that if there exists a set of vertices~$S$ with small conductance, then there exists a set of weak vertices~$T$ whose volume is at least a constant fraction of the volume of~$S$. We defer the proof of this technical lemma to the appendix.

  \begin{restatable}{lemma}{thmspeccorollary}
    \label{thm:spec_corollary}
    Let $S \subset V$ be such that $\vol{S} \le \vol{\bar{S}}$ and $\setcond{S} \leq \delta$.
    Then, for any $\ell \in \setn$ and any $0 < \theta \leq 1/10$, there exists a set $T \subseteq S$  such that $\vol{T} \geq \theta \vol{S}$ and for every $v \in T$, it holds that
    $\normb{\Walk^\ell(v,\cdot) - \statdist}^2 > \frac{1}{80 m^7} (1 - 4 \delta)^{2 \ell}$.
  \end{restatable}
  \proofref{proof:thm_spec_corollary}

  We can use \cref{thm:spec_corollary} to separate weak vertices from the remaining graph.

  \begin{lemma}
    Let $G=(V,E)$ be a graph.
    If the volume of weak vertices in $G$ is at most $(1/100) \epsilon m$, then there is a partition of $V$ into $P \cup \bar{P}$ such that $\vol{P} \leq \epsilon m / 10$ and $\graphcond{G[\bar{P}]} \geq \Phi^2 / 256$.
    \label{thm:small-volume-weak}
  \end{lemma}
  \begin{proof}
    We partition the graph recursively into two sets $(P, \bar{P})$.  At the beginning, $P_0 =
    \emptyset$ and $\bar{P}_0 = V$.  As long as there is a cut $(C_i, \bar{C}_i)$ in $\bar{P}_{i-1}$
    in step $i$ with $\vol{C_i} \le \vol{\bar{C}_i}$ and $\edgecut{E}{C_i}{\bar{C}_i} / \vol{C_i} \leq \Phi^2/256$, we set $P_i = P_{i-1} \cup C_i$ and $\bar{P}_i = V \minus P_i$. We continue this until
    we don't find such a cut or the condition $\vol{P_{i+1}} \le \vol{\bar{P}_{i+1}}$ would be violated. The number of edges going across the cut $(P,\bar{P})$ is at most $\sum_i |\edgecut{E}{C_i}{\bar{C}_i}|$. Therefore, $|\edgecut{E}{P}{\bar{P}}| \leq \tfrac{\Phi^2}{256} \sum_i \vol{C_i} \leq \tfrac{\Phi^2}{256} \vol{P}$.

    Now, assume that $\vol{P} > (1/10) \epsilon m$. 
    \Cref{thm:spec_corollary} implies that there exists $P' \subseteq P$ such that $\vol{P'} \geq \frac{1}{10} \vol{P} > \epsilon m / 100$ (where $\theta = 1/10$) and for all $v \in P'$ we have
      $\normb{\Walk^\ell(v, \cdot) - \statdist}
      > \frac{1}{80m^7} (1 - 4 \Phi^2 / 256)^{2 \ell}
      > \frac{1}{80m^{10}}$.
      This means that $P'$ contains only weak vertices and has volume at least $\epsilon m/100$, which contradicts our assumption that the volume of weak vertices in $G$ is at most $\epsilon m/100$. Therefore, $\vol{P} \le \epsilon m/10$ when the partitioning terminates.
      Hence $\graphcond{G[\bar{P}]} \geq \Phi^2/256$.
  \end{proof}

  Finally, the following lemma states that few edge modifications in a graph with separated weak vertices are sufficient to make it a graph with high conductance.

  \begin{lemma}[{\citep[Lemma 9]{LiTes15}}]
    \label{thm:patch_graph}
    Let $G=(V,E)$ be a graph.
    If there exists a set $P \subseteq V$ such that $\vol{P} \leq \epsilon m / 10$ and the subgraph
    $G[V \minus P]$ is a $\Phi'$-expander, then there exists an algorithm that modifies at most
    $\epsilon m$ edges to get a $\Phi' / 3$-expander $G' =(V, E')$.
    \label{thm:close}
  \end{lemma}

  Combining the results on the separation of weak vertices and patching the graph (\cref{thm:spec_corollary,thm:small-volume-weak,thm:patch_graph}) and approximating the endpoint distribution (\cref{thm:approx_end_dist,thm:approx_summand}), we prove the soundness of the algorithm.

  \begin{lemma}[Soundness]
    Let $G(V,E)$ be a graph. If $G$ is $\epsilon$-far from having conductance at least $\Phi^2 / 768$, then, with probability at least $2/3$, each vertex in $G$ returns $\Reject$ when it runs \cref{alg:cond-tester}.
  \end{lemma}
  \begin{proof}
    First we note that if the volume of weak vertices is less than $\epsilon m/100$, then by \cref{thm:small-volume-weak,thm:close}, the graph is $\epsilon$-close to having conductance at least $\Phi^2 / 768$. Therefore, the volume of weak vertices is at least $\epsilon m/100$. Each vertex $v$ is contained in $S$ with probability $\Theta(\degr{v} / \epsilon m)$. Hence, the expected number of weak vertices that are present in the sample $S$ is at least $100$. Therefore, with probability at least $9/10$, at least one weak vertex is sampled in $S$.

    If $\Walkex[v,u]{\ell} < m^{-2}$ for some $v \in S, u \in V$, then with probability at least $9/10$, $\widehat{\Walk}^\ell_{v,u} < 2 m^{-2}$ by \cref{thm:approx_end_dist}. In this case, the algorithm will reject in line~\ref{lin:small_W_reject} of \cref{alg:random-walk}.
    If $\Walkex[v,u]{\ell} \geq m^{-2}$ for all $v \in S, u \in V$, then with probability at least $9 / 10$, it holds that $\left| \normb{\Walkex[v,\cdot]{\ell} - \statdist} - s_v \right| \leq 3m^{-19}$ for every $v \in S$ by \cref{thm:approx_summand}. Since at least one vertex $v \in S$ is weak, \ie, $\normb{\Walk^\ell(v,\cdot) - \statdist} > 6m^{-15}$, the algorithm rejects in line~\ref{line:conductance_reject} of \cref{alg:cond-tester}.
  \end{proof}

  \subsection{Unknown Size of the Graph}
  \label{sec:unknown-size-graph}

  We describe how to get rid of the assumption that the size $n$ of the graph $G$ is known to the tester if $G$ is connected. Note that without any prior knowledge of $G$, no distributed tester can distinguish between a graph with conductance $\Phi$ and two distinct copies of it (the latter graph has conductance~$0$ and is $\epsilon$-far from being a graph with conductance $\Phi^c$ for $\epsilon < \Phi^c / 2$, $c \geq 1$).

  First, we describe a slightly simpler version of the final algorithm. In the setting of the simpler algorithm, we mark a single vertex that will initiate the test and will also give the final answer of the tester. We call this vertex the maintainer (of the graph). The algorithm can be easily adapted to the CONGEST model.

  Let $v \in V$ be a fixed vertex. The algorithm either makes $n$ available at all vertices and runs \cref{alg:cond-tester} afterwards or $v$ rejects because $G$ does not have conductance $\Phi$. If $G$ has conductance $\Phi$, the algorithm never rejects.

  We start with an initial set $S = \{ v \}$ that is grown in two phases. In the first phase, we extend $S$ to $S \cup \neighbors{S}$ as long as $\setcond{S} \geq \Phi$. In particular, $v$ starts a BFS and in every round, the vertices in the last level report their degree and the number of neighbors outside of $S$ to their parents. Similar to \cref{alg:cond-tester}, these are aggregated and sent to $v$ along the edges of the BFS tree. If $\setcond{S} < \Phi$ for the first time, the algorithm proceeds to the second phase. It continues the BFS for $- \log(\vol{S}) / \log (1-\Phi)$ rounds and stops. If any vertex in the graph notices a neighbor that is not in $S$ after these rounds, then $S \neq V$ and the algorithm rejects. Otherwise, we have obtained the value of $n = |S|$ that can be sent to all vertices, and we continue by executing \cref{alg:cond-tester}.

  \begin{restatable}{lemma}{thmunknownn}
    \label{thm:unknown_n}
    Let $G=(V,E)$ be a graph and $\Phi \in [0,1]$. There is an algorithm that computes $n$ if $G$ has conductance at least $\Phi$. 
    Otherwise, it either computes $n$ or rejects.
    The round complexity is $\OO(\log m / \log (1-\Phi))$.
  \end{restatable}
  \begin{proof}
    It is easy to see that if the algorithm explores the whole graph, it computes $n$ correctly, and else it rejects. \Wlog, let $G$ have conductance $\Phi$. Let $S_i$ be the set $S$ after $i$ rounds and let $\bar{S}_i = V \minus S_i$. We denote the last round of the first (second) phase by $k$ ($\ell$).

    In the first phase, we have that $\vol{S_i} \geq (1+\Phi) \cdot \vol{S_{i-1}}$ for every round $i$ and by induction, $k \leq \log \vol{S_k} / \log (1+\Phi) \leq \log m / \log (1+\Phi)$. We also have that $\vol{S_k} \geq m/2 \geq \vol{\bar{S}_k}$ because $G$ has conductance $\Phi$. In the second phase, we have that $\vol{\bar{S}_i} \leq (1-\Phi) \cdot \vol{\bar{S}_{i-1}}$ for every round $i$. By induction, $\ell - k \geq - \log m / \log (1-\Phi) \geq \log \vol{\bar{S}_k}^{-1} / \log (1-\Phi)$ implies that that $\vol{\bar{S}_\ell} = 0$. Therefore, the algorithm has explored the whole graph.
    Clearly, $\ell \in \OO(\log m / \log (1-\Phi))$.
  \end{proof}

  To transform the algorithm into a tester in the CONGEST model, we start with each vertex being a maintainer initially. In every round every vertex chooses the vertex with the smallest id it has ever received a message from to be the maintainer and it forwards only this vertex' messages (the latter maintains the congestion bound). At the end of the algorithm, if $G$ has conductance $\Phi$, then there is only one maintainer (the vertex with the smallest id) and the algorithm continues by executing \cref{alg:cond-tester}. Otherwise, there might be multiple vertices that are still maintainers. However, none of these vertices has explored the whole graph, so all of them send a broadcast message to reject.

  \section{Lower Bound}
  \label{sec:lower_bound}

  In this section, we prove a lower bound of $\Omega(\log (n+m))$ on the round complexity for testing the conductance of a graph in the LOCAL model regardless of how the final decision of the tester is derived from the single votes of the vertices.
  
  For any $v\in V$, the \kdisc of $v$, denoted by $\fdisc{G}{v}$, is defined as the subgraph that is induced by the vertices that are at distance at most $k$ to $v$ without the edges between vertices at distance exactly $k$, and it is rooted at $v$.  We refer to the isomorphism type of $\fdisc{G}{v}$, \ie, the set of all rooted graphs isomorphic to $\fdisc{G}{v}$, by $\fdisci{G}{v}$. Let $\girth{G}$ denote the length of the shortest cycle in $G$. We need the following two lemmas to obtain the distribution over graphs to prove the lower bound.
  \begin{lemma}[{\citep{LubRam88}; \cf \citep[Section 16.8.3]{NauCom12}}]
    \label{thm:ramanujan_graphs}
    For every $n' \in \setn$ and every $d' \in \setn$ there exists a $d$-regular graph $G$ of size $n$ such that $G$ has conductance $\graphcond{G} = 1/\sqrt{2d}$ and girth $2 \log n / \log d$, and $n \geq n'$, $d \geq d'$.
  \end{lemma}
  
  The second lemma states that we can sparsify an arbitrary cut $\edgecut{E}{V_1}{V_2}$ in a $d$-regular graph with girth $3k$ without changing $\fdisci{G}{v}$ for any $v \in V$.
  In particular, it states that we can remove two edges in the cut and add them somewhere else, or the cut has size $\poly(d^k)$ only.
 	\begin{lemma}[{\citep[Lemma 8]{FicCon15}}]\footnotemark{}
  \footnotetext{The statement here is obtained as a special case by observing that we can assume $L = 1$ and $\lambda = 0$ in {\citep[Lemma 8]{FicCon15}}.}
 		\label{thm:rewiring}
 		Let $G=(V, E)$ be a $d$-regular graph with $\girth{G} \geq 3k$ for $k \geq 2$ and let $V_1 \cupdot V_2 = V$ be a partitioning of $V$.
 		Then either there exists a graph $H=(V, F)$ such that
    \begin{inenum}
      \item $\girth{H} \geq 3k$,
      \item $|F \cap (V_1 \times V_2)| \leq |E \cap (V_1 \times V_2)| - 2$, and 
      \item $\fdisci{H}{ w} = \fdisci{G}{ w} \forall w \in V$, or $\edgecut{E}{V_1}{V_2}
 			\leq 6d^{3k}$.
   \end{inenum}
 	\end{lemma}

  To prove the lower bound, we use an auxiliary model we call the ISO-LOCAL model. In this model, the input $I(\cdot)$ is empty but an additional oracle provides every vertex $v$ with the ability to construct $\fkdisci{r}{G}{v}$ in round $r$ if it knows $\fkdisci{r-1}{G}{u_i}$ of its neighbors $u_1, \ldots, u_{\degr{v}}$.  It should be noted that the ISO-LOCAL model is not a DCM due to the additional oracle.
  \begin{definition}[ISO-LOCAL model]
    Let $\dcm{G}{p_G}{I}$ be a DCM instance such that $I$ maps the whole support to the empty string.
    In addition to sending and receiving messages, in every round $r$ every vertex $v$ is provided
    access to a function $e_{r,v} : (\setn \cup \{ \star \})^r \times (\setn \cup \{ \star \})^r
    \rightarrow \{ 0, 1 \}$ during the local computation phase.
    The value of $e_{r,v}((i_1, \ldots, i_{r'}), (j_1, \ldots, j_{r''}))$ is $1$ iff $p'_v(i_1,
    \ldots, i_{r'}) = p'_v(j_1, \ldots, j_{r''})$, where
    \begin{equation*}
      p'_v(i_1, \ldots, i_r) \defeq
      \begin{cases}
      	v                                          & \text{ if } i_r = \star  \\
      	p_{p'_v(\star, i_1, \ldots, i_{r-1})}(i_r) & \text{ otherwise} \fstop
      \end{cases}
     \end{equation*}
     The instance $\dcm{G}{p_G}{I}$ equipped with such an oracle is called ISO-LOCAL.
   \end{definition}
   In other words, $p'_v(\cdot)$ takes a path of length at most $r$ that starts at $v$ and that is defined by a sequence of port numbers as input.  Then, it maps the path to its endpoint in $V$.
   Finally, $e_{r,v}(\cdot)$ tells whether two such paths end at the same vertex.

  It is a basic observation that a distributed algorithm can only depend on information that has reached it until the moment it performs the computation in question.
  \begin{lemma}[folklore; {\cf \citep[Section~2]{LinLoc92}}]
    \label{thm:distributed_communication_depth}
    Let $\dna{\dalg{A}}{O}$ be a DNDA. After $r$ rounds, the state of $\dalgon{A}{v}$ may depend only on $\degr{v}$, $I(v)$, the state of $\dalgon{A}{u}$ at time $r - \dist{v}{u}$ for vertices $u$ with $\dist{v}{u} < r$ and the random coins of $\dalg{A}$.
  \end{lemma}

  \subsection{Proof of the Lower Bound}

  Let $G = (V, E)$ be an expander graph obtained from applying \Cref{thm:ramanujan_graphs} and let $k = \Theta(\log n)$. Observe that if a graph is $d$-regular and it has girth $3k$, then all its \kdiscs are pairwise isomorphic.  In particular, all \kdiscs are full $d$-ary trees of depth $k$.

  We will prove that a distributed algorithm $\dna{\dalg{A}}{O}$ with round complexity $r$ in the ISO-LOCAL model decides based on the set of views $\fkdisci{r}{G}{v}$ that the different instances of $\dalg{A}$ have (see \Cref{thm:iso_local_only_iso_type}).  Using \Cref{thm:rewiring}, it will be easy to come up with a graph $H$ that is a bad expander but whose \kdiscs are isomorphic to the ones of $G$.  This implies a lower bound of $k = \Theta(\log n)$ for testing conductance in the ISO-LOCAL model (see \Cref{thm:lower_bound_iso_local}).  Finally, we prove (in the appendix) that a lower bound on the round complexity of a tester in the ISO-LOCAL model implies the same bound in the LOCAL model. Actually, we prove the contrapositive: a tester in the LOCAL model implies a tester in the ISO-LOCAL model (see \Cref{thm:lower_bound_implication}).

  \begin{restatable}{lemma}{thmisolocalonlyisotype}
    \label{thm:iso_local_only_iso_type}
    Let $\dna{\dalg{A}}{O}$ be a deterministic DNDA in the ISO-LOCAL model.
    The output of $\dalgon{A}{v}$ depends only on $\fkdisci{r}{G}{v}$ and the port numbering $(p_v)_{v \in V}$.
  \end{restatable}
  
  \begin{proof}
    Instead of analyzing $\dna{\dalg{A}}{O}$, we will analyze a canonical algorithm $\dna{\dalg{B}}{O}$ that simulates $\dna{\dalg{A}}{O}$ depending only on $\fkdisci{r}{G}{v}$. Employing $\dalg{B}$, we prove the following statement by induction: After the local computation phase of round $r$, the state of $\dalgon{A}{v}$ depends only on $\fkdisci{r}{G}{v}$.
  
    The first local computation phase of $\dalgon{A}{v}$ can only depend on the port numbering and $I(v)$ (the empty string). Therefore, $\dalgon{B}{v}$ can simulate the execution of the first round of $\dalgon{A}{v}$.
  
    Let the current round be $r > 1$. Algorithm $\dalgon{B}{v}$ maintains a rooted graph $H_{v}$ that resembles $\fkdisci{r}{G}{v}$. The adjacency lists of $H_{v}$ are ordered according to $(p_v)_{v \in V}$. Let $H_v(r)$ be the value of $H_v$ after the computation phase of round $r$. In the send phase, vertex $v$ sends $H_v(r)$ to each of its neighbors. In the receive phase, vertex $v$ receives graphs $H_{u_1}(r), \ldots, H_{u_{\degr{v}}}(r)$ from its neighbors $u_1, \ldots, u_{\degr{v}}$. In the subsequent computation phase of round $r+1$, vertex $v$ extends $H_v(r) = \fkdisci{r}{G}{v}$ to $\fkdisci{r+1}{G}{v} = H_v(r+1)$ by querying $e_{r,v}$ on all pairs of vertices of $V(H_v(r)) \cup V(H_{u_1}(r)) \cup \ldots \cup V(H_{u_{\degr{v}}}(r))$ to identity vertices and patching the different views together.
  
    Note that $H_v(r+1)$ also provides the isomorphism type of $\fkdisci{r - \dist{v}{u}}{G}{u}$ for every vertex $u$ at distance at most $r$ from $v$. Since the adjacency lists of $H_v$ are ordered according to the port numbering, it is also possible to reconstruct $e_{r-\dist{v}{u},u}(\cdot)$. By the induction hypothesis, $\dalgon{B}{v}$ can now simulate round $r - \dist{v}{u}$ of $\dalgon{A}{u}$ for every such $u$. By \Cref{thm:distributed_communication_depth}, this is enough to simulate the local computation phase of round $r$ of $\dalgon{A}{u}$.
  \end{proof}

  We use the lemma to show that there is no tester for conductance in the ISO-LOCAL model.  
  
  \begin{proposition}
    \label{thm:lower_bound_iso_local}
    Let $G=(V,E)$ be a graph on $n$ vertices, and let $\Phi > 0$ be any constant. Any algorithm for testing if $G$ has conductance at least $\Phi$ or is $\epsilon$-far from having conductance at least $c\Phi^2$ (for a constant $c$) in the ISO-LOCAL model that succeeds with probability $2/3$ requires $\Omega(\log n )$ rounds of communication.
  \end{proposition}
  \begin{proof}
    Let $G=(V,E)$ be a $d$-regular graph provided by \Cref{thm:ramanujan_graphs} and set $k = \tfrac{1}{3} \log_d \left( \tfrac{c \Phi^2 - \epsilon}{6} d n \right)$.
    
    \Wlog assume that $n$ is even, and let $S \subset V$ be a set of size $n/2$. Apply \Cref{thm:rewiring} (with $V_1 = S$ and $V_2 = V \minus S$) repeatedly to $G$ until $|\edgecut{E}{S}{V \minus S}| \leq 6d^{3k}$ holds. Let $H=(V,E')$ be the resulting graph. We have that $|E'(S, V \minus S)| \leq (c \Phi^2 - \epsilon) d n$, and $\vol{S} = nd / 2$. Therefore, $H$ is $\epsilon$-far from having conductance $c \Phi^2$. Let $\mathcal{D}_G$ ($\mathcal{D}_H$) be the uniform distribution over all ISO-LOCAL models $\dcm{G}{p_G}{I}$ ($\dcm{H}{p_H}{I}$) such that $p_G$ ($p_H$) ranges over all possible mappings, \ie, port numberings.
  
    We use Yao's principle to prove the lower bound.  Let $\dna{\dalg{A}}{O}$ be a tester for conductance that has round complexity smaller than $k$ in the ISO-LOCAL model.  Since $G$ is $d$-regular and $\girth{G} \geq 3k$, $\fkdisci{k}{G}{v}$ is a full $d$-ary tree of depth $k$ for every $v \in V$.  For any pair $u,v \in V$, we have that $\fdisci{G}{u}$ is equal to $\fdisci{H}{v}$ by \Cref{thm:rewiring}.  Since the port numberings of two vertices are independent of each other, $(p_v)_{v \in V}$ is a valid port numbering for $G \in \mathcal{D}_G$ iff it is valid for $H \in \mathcal{D}_H$.  By \Cref{thm:iso_local_only_iso_type}, $\dna{\dalg{A}}{O}$ cannot distinguish between $G$ and $H$.
  \end{proof}
  
  To complete the proof of the lower bound, we show (in the appendix) that each vertex in the graph in the ISO-LOCAL model can choose an id randomly, and with high probability no two ids will be identical.
  
  \begin{restatable}{proposition}{thmlowerboundimplication}
    \label{thm:lower_bound_implication} Let $\dna{\dalg{A}}{O}$ be a randomized tester in the LOCAL model that succeeds with probability $p$. Then, there is a randomized tester $\dna{\dalg{B}}{O}$ in the ISO-LOCAL model that succeeds with probability at least $p - o(1)$, and has the same round complexity.
  \end{restatable}
  \proofref{proof:thm_lower_bound_implication}

  \section{Open Problems}
  \label{sec:open_problems}
  
  In the case of one-sided distributed testers, it is natural to define the acceptance rule $O(\cdot)$ of a distributed tester such that all vertices have to accept or at least one vertex has to reject.  This is because in the case of rejection, the tester is required to observe a witness.  However, for two-sided testers no such requirement exists.  Requiring that all vertices either accept or reject simultaneously, which can be satisfied by a slightly modified version of \cref{alg:cond-tester}, seems to be quite strong.  On the other hand, it might not always be possible to obtain a lower bound that is independent of the acceptance rule as in \cref{thm:lower_bound_nn}. To this end, it would be interesting to compare the power of different rules.
  
  \section*{Acknowledgments}
  We would like to thank Gopal Pandurangan for pointing out related work \cite{DasDis13, MolDis17}, and we would like to thank Pan Peng for inspiring discussions on spectral graph theory.
  We are grateful for the helpful comments of anonymous reviewers.
  
  \newpage
    \bibliographystyle{plainnat}
    \bibliography{literature_cleaned.bib} %
    \clearpage
  
  \appendix
  
  \section{Proofs from \Cref{sec:upper_bound}}
  \label{sec:app-2}
  
  \subsection{\Cref{alg:bfs,alg:sum}}
  \label{subsec:alg}
  
    \begin{algorithm}
      \caption{Construct BFS tree}
      \label{alg:bfs}
      \begin{algorithmic}[1]
        \Procedure{BFS}{$G$, $D$}
          \CompAllVert{$v$}
            \State $T_v \gets (v, \cdot)$ \Comment{Set root to itself, parent to empty}
            \State $minid \gets v$
            \State send $(v, v)$ to every neighbor $u \in \neighbors{v}$
          \EndCompAllVert
          \For{\DistNumRoundsAllVert{$D$}{$w$}}
            \State $R_v \gets \{ (v', u') \text{ received} \sep u' \in \neighbors{w} \}$
            \State $(v, u) \gets \arg\min_{(v',u') \in R_v} v'$
            \If{$T_w = (\cdot)$ or $v' < minid$}
              \State $T_w \gets (v, u)$ \Comment{Set root to $v$, parent to $u$}
              \State send $(v, w)$ to all neighbors $\neq u$
            \EndIf
          \EndFor
        \EndProcedure
      \end{algorithmic}
    \end{algorithm}
  
    \begin{algorithm}
      \caption{Aggregate sum of vertex values and propagate it to all vertices}
      \label{alg:sum}
      \begin{algorithmic}[1]
        \Require $\forall v : v \text{ has local information } f(v)$
        \Ensure $\forall v : v \text{ has information } \sum_{u \in V} f(u)$
        \Procedure{AggregateSum}{$G$, $D$, $f : V \rightarrow \setr$}
          \For{\DistNumRoundsAllVert{$D$}{$v$}}
            \If{$v$ received partial sums $s_u$ from all its children $u$ in BFS tree}
              \State $s_v \gets f(v) + \sum_{u} s_u$
              \State send $s_v$ to parent in BFS tree
            \EndIf
          \EndFor
          \CompVert{root $r$ of BFS tree}
            \State send total sum $s = \sum_v f(v)$ to all children
          \EndCompVert
          \For{\DistNumRoundsAllVert{$D$}{$v$}}
            \If{$v$ received total sum $s$ from its parent}
              \State send $s_v$ to all children in BFS tree
            \EndIf
          \EndFor
          \State \Return $s_v$ \Comment{consider $s_v$ to be the output of the algorithm}
        \EndProcedure
      \end{algorithmic}
    \end{algorithm}
    
    \FloatBarrier
  
  \subsection{Proof of \Cref{thm:approx_end_dist}}
  \label{subsec:complete-app}
  
    \thmapproxenddist*
    \begin{proof}
      \label{proof:approx_end_dist}
      We have $E[\widehat{\Walk}^\ell_{v,u}] = \Walkex[v,u]{\ell}$.
      By Hoeffding's inequality, it holds that
      \begin{equation}
        \label{eq:approx_end_dist-hoeff}
        \Pr[|\widehat{\Walk}^\ell_{v,u} - E[\widehat{\Walk}^\ell_{v,u}]| \geq m^{-10}] \leq 2 \exp \left( -\frac{N}{3 m^{40}} \right) \leq m^{-10} \fstop
      \end{equation}
  
      Condition on $|\widehat{\Walk}^\ell_{v,u} - E[\widehat{\Walk}^\ell_{v,u}]| < m^{-20}$, which happens with probability at least $1 - 1 / m^{10}$.
      If $\widehat{\Walk}^\ell_{v,u} < m^{-2}$, then
      \begin{equation*}
        \Walkex[v,u]{\ell} = E[\widehat{\Walk}^\ell_{v,u}] < \widehat{\Walk}^\ell_{v,u} + m^{-20} = m^{-2} + m^{-10} < 2 m^{-2} \fstop
      \end{equation*}
      Let $\pi' = \onevector{v}$.
      We bound $\normb{\Walk^\ell \vecto{\pi'} - \statdist}$ from below.
      \begin{equation*}
        \normb{\Walk^\ell \vecto{\pi'} - \statdist}
        \geq |\Walkex[v,u]{\ell} - \degr{u}/2m|
        \geq -(2m^{-2} - 1/(2m))
        \geq 1 / (4m) \fstop
      \end{equation*}
      By the contrapositive of \cref{thm:high-cond-close}, $G$ has conductance less than~$\Phi$.
    \end{proof}

  \subsection{Proof of \Cref{thm:spec_corollary}}
  \label{subsec:proof-sound}
  \label{proof:thm_spec_corollary}
  
  The \emph{(normalized) Laplacian} becomes useful when studying cuts in a graph.
  \begin{definition}[Normalized Laplacian]
    Let $G=(V,E)$ be a graph.  The normalized Laplacian is defined by $\Nlap \defeq \Id -
    \Degr^{-1/2} \Adj \Degr^{-1/2}$.
  \end{definition}
  
  The following decomposition of the walk matrix turns out to be useful in this context.
  \begin{lemma}[{\citep[Lemma 12.2]{LevMar09}}]\footnotemark{}
    \footnotetext{In \citep{LevMar09}, the result is stated for a right stochastic walk matrix and its \emph{right} eigenbasis that is orthonormal with respect to the non-standard inner product $\innerprod{f}{g}_\pi = \sum f(x) g(x) \statdist(x)$.
      The statement here is adapted to our notation.}
    \label{thm:walk_matrix_decomposition}
    Let $G$ be a graph.  The walk matrix $\Walk$ can be decomposed as
    \begin{equation*}
      \frac{\Walk^\ell(u,v)}{\statdist(v)} = 1 + \sum_{i=2}^{n} \Nwalkew_i^\ell \frac{\Nwalkev_i(u)
      \Nwalkev_i(v)}{\sqrt{\degr{u} \degr{v}}} \fstop
    \end{equation*}
  \end{lemma}
  
  We would like to prove that if the conductance of some set $S$ of vertices is small, then a constant fraction of the volume of $S$ belongs to some (basically) weak vertices. The following statement is a preliminary version of the result that we aim for. It proves the existence of a single vertex with some (unknown) volume only.

  \begin{lemma}
    \label{thm:exists_one_weak_vertex}
    Let $S \subset V$ be such that $\vol{S} \le \vol{\bar{S}}$ and $\setcond{S} \leq \delta$.
    Then, for any $\ell \in \setn$, there exists a vertex $v \in S$ such that
    \begin{equation*}
      \normb{\Walk^\ell(v,\cdot) - \statdist}^2 > \frac{1}{16 m^7} (1 - 4 \delta)^{2 \ell}
    \end{equation*}
  \end{lemma}
  \begin{proof}
    We will argue that $\frac{1}{s} \sum_{x \in S} \normb{\Walk^\ell(x,\cdot) - \statdist}^2 > (1 - 4 \delta)^{2 \ell} / (16 m^7)$ and apply an averaging argument to conclude.

    Assume for the moment that
    \begin{equation}
      \label{eq:spec_bound_x_alpha_prelim}
      \sum_{\substack{i \geq 2\\ \Nwalkew_i > \tau}}^{n} \left( \sum_{x \in S} \sqrt{\degr{x}} \Nwalkev_i(x) \right)^2
      \geq \frac{\vol{S}}{4}, \text{ where $\tau = (1 - 4\delta)$}.
    \end{equation}

    Then, the following calculation concludes the proof:
    \begin{align*}
      & \frac{1}{s} \sum_{x \in S} \normb{\Walk^\ell(x,\cdot) - \statdist}^2 \\
      = &\frac{1}{s} \sum_{x \in S} \hnormb{\frac{\Degr}{2m} \left( 2m D^{-1} \Walk^\ell(x,\cdot) - \vecto{1} \right)}^2 \\
      = & \frac{1}{4m^2s} \sum_{x \in S} \hnormb{ \Degr \sum_{i \geq 2}^{n} \Nwalkew_i^\ell \frac{\Nwalkev_i(x)}{\sqrt{\degr{x}}} \Degr^{-1/2} \Nwalkev_i }^2 \text{, by \cref{thm:walk_matrix_decomposition}} \\
      \geq & \frac{1}{4 m^2 n} \hnormb{ \sum_{x \in S} \frac{\Deg^{1/2}}{s \cdot \degr{x}} \sum_{i \geq 2}^{n} \Nwalkew_i^\ell \sqrt{\degr{x}} \Nwalkev_i(x) \Nwalkev_i }^2 \text{, by Jensen's inequality} \\
      \geq & \frac{1}{4 m^2 n^3 s^2} \tau^{2 \ell} \sum_{\substack{i \geq 2\\ \Nwalkew_i > \tau}}^{n}
      \left( \sum_{x \in S} \sqrt{\degr{x}} \Nwalkev_i(x) \right)^2 \normb{ \Nwalkev_i }^2 \\
      \geq & \frac{1}{4 m^2 n^3 s^2} \tau^{2 \ell} \frac{\vol{S}}{4} \text{, by \cref{eq:spec_bound_x_alpha_prelim}} \\
      \geq & \frac{1}{16 m^7} (1 - 4 \delta)^{2 \ell} \text{, w.l.o.g. $\vol{S} \geq 1$} \fstop
    \end{align*}

    The remaining calculation is similar to the proof of \citep[][Lemma 3.5]{KalExp11}. We prove \cref{eq:spec_bound_x_alpha_prelim}. Let $\vecto{u} = D^{1/2} \vecto{1}_S$. Denote $\alpha_i \defeq \hinnerprod{\Degr^{1/2} \vecto{1}_S}{\Nwalkev_i}$. Representing $\vecto{u}$ in the orthonormal eigenbasis $\{ \Nwalkev_i \}_{i \in \onerange{n}}$ of $\Nwalk$, we have
    \begin{equation*}
      \vecto{u}
      = \sum_{i=1}^{n} \alpha_i \Nwalkev_i \fstop
    \end{equation*}

    By the definition of the normalized Laplacian,
    \begin{align}
      \label{eq:spec_lap}
      \vecto{u}^\Tr \Nlap \vecto{u}
        & = \vecto{u}^\Tr I \vecto{u} - \vecto{u}^\Tr \Nwalk \vecto{u} \fstop
    \end{align}

    Observe that
    \begin{gather}
      \label{eq:spec_uu}
      \vecto{u}^\Tr I \vecto{u}
      = \normb{\vecto{u}}^2
      = \sum_{i}^{n} \alpha_i^2\ \\
      \label{eq:uu_vol}
      \normb{\vecto{u}}^2
      = \sum_{i \in S} \sqrt{\degr{i}}^2
      = \vol{S} \fstop
    \end{gather}
    The second term of the right-hand side of \cref{eq:spec_lap} is equal to
    \begin{align}
      \label{eq:spec_uNu}
      \vecto{u}^\Tr \Nwalk \vecto{u}
      = \left( \sum_{i=1}^{n} \alpha_i \Nwalkev_i^\Tr \right) \Nwalk \left( \sum_{i=1}^{n} \alpha_i \Nwalkev_i \right)
      & = \left( \sum_{i=1}^{n} \alpha_i \Nwalkev_i^\Tr \right) \left( \sum_{i=1}^{n} \alpha_i \Nwalkew_i \Nwalkev_i \right) \nonumber \\
      & = \left( \sum_{\substack{i,j=1\\ i=j}}^{n} \alpha_i^2 \Nwalkew_i \Nwalkev_i^\Tr \Nwalkev_j \right)
        + \left( \sum_{\substack{i,j=1\\ i \neq j}}^{n} \alpha_i^2 \Nwalkew_i \Nwalkev_i^\Tr \Nwalkev_j \right) \nonumber \\
      & = \sum_{i=1}^{n} \alpha_i^2 \Nwalkew_i + 0 \fstop
    \end{align}

    Combining \cref{eq:spec_lap,eq:spec_uu,eq:spec_uNu}, we get that
    \begin{align}
      \label{eq:spec_uLu_a}
      \vecto{u}^\Tr \Nlap \vecto{u}
        & = \vecto{u}^\Tr I \vecto{u} - \vecto{u}^\Tr \Nwalk \vecto{u}
        = \normb{\vecto{u}}^2 - \sum_{i=1}^{n} \alpha_i^2 \Nwalkew_i \fstop
    \end{align}

    On the other hand,
    \begin{equation}
      \label{eq:spec_uLu_b}
      \vecto{u}^\Tr \Nlap \vecto{u}
      = \sum_{(i,j) \in E} \left( \frac{\vecto[i]{u}}{\sqrt{\degr{i}}} - \frac{\vecto[j]{u}}{\sqrt{\degr{j}}} \right)^2
      = \sum_{(i,j) \in \edgecut{E}{S}{\bar{S}}} (1-0)
      \leq \delta \vol{S} \fstop
    \end{equation}
    
    \Cref{eq:uu_vol,eq:spec_uLu_a,eq:spec_uLu_b} imply
    \begin{equation}
      \label{eq:spec_alpha_bound}
      \sum_{i \in \onerange{n}} \alpha_i^2 \Nwalkew_i \geq (1 - \delta) \vol{S} \fstop
    \end{equation}

    Let $H$ be the eigenvalues $\Nwalkew_i > 1 - 4\delta$, and define $x \defeq \sum_{\Nwalkew \in H} \alpha_i^2$.
    Rewriting \cref{eq:spec_alpha_bound}, we get that
    \begin{alignat}{3}
      \label{eq:spec_bound_large_ev}
      & x + \left( \sum_{i \in \onerange{n}} \alpha_i^2 - x \right) \left( 1 - 4\delta \right)
        && \geq (1 - \delta) \vol{S}
        && ~ \nonumber \\
      \Leftrightarrow \quad & 4 \delta x + \vol{S} (1 - 4 \delta)
        && \geq (1 - \delta) \vol{S}
        && \text{, by \cref{eq:spec_uu}} \nonumber \\
      \Leftrightarrow \quad & x
        && \geq \frac{3\vol{S}}{4}
        &&~
    \end{alignat}

    Note that
    \begin{equation*}
      \label{eq:spec_alpha_one}
      \alpha_1
      = \hinnerprod{\Degr^{1/2} \vecto{1}_S}{\sqrt{2m}^{-1}\Degr^{1/2} \vecto{1}}
      = \frac{1}{\sqrt{2m}} \sum_{i \in S} \degr{i}
      \leq \frac{\vol{S}}{\sqrt{2\vol{S}}}
      = \sqrt{\frac{\vol{S}}{2}} \fstop
    \end{equation*}

    Therefore, we have that
    \begin{equation}
      \label{eq:spec_bound_x_alpha}
      \sum_{\substack{i \geq 2\\ \Nwalkew_i > \tau}}^{n} \left( \sum_{x \in S} \sqrt{\degr{x}}\Nwalkev_i(x) \right)^2
      = x - \alpha_1^2
      \geq \frac{\vol{S}}{4} \fstop
    \end{equation}
  \end{proof}
  
  Actually, our goal is to get a result that is a bit stronger than \cref{thm:exists_one_weak_vertex}. However, it follows from \cref{thm:exists_one_weak_vertex} as \citep[][Lemma 3.6]{KalExp11} follows from \citep[][Lemma 3.5]{KalExp11}. It states that even if we exclude some vertices that account for a small fraction of the total volume of $S$, there exists a basically weak vertex.

  \begin{lemma}
    \label{thm:exists_set_weak_vertices}
    Let $T \subseteq S \subset V$ be such that $\vol{S} \le \vol{\bar{S}}$, $\setcond{S} \leq \delta$ and
    $\vol{T} = (1 - \theta) \vol{S}$ for some $0 < \theta \leq 1/10$.  Then, for any $\ell \in
    \setn$, there exists a vertex $v \in T$ such that
    \begin{equation*}
      \normb{\Walk^\ell(v,\cdot) - \statdist}^2 > \frac{1}{80 m^7} (1 - 4 \delta)^{2 \ell}
    \end{equation*}
  \end{lemma}
  \begin{proof}
    Let $\vecto{u}_S \defeq D^{1/2} \vecto{1}_S$ and $\vecto{u}_T \defeq D^{1/2} \vecto{1}_T$.
    Let $\alpha_i \defeq \innerprod{\vecto{u}_S}{\Nwalkev_i}$ and $\beta_i \defeq \innerprod{\vecto{u}_T}{\Nwalkev_i}$.
    \Cref{eq:spec_bound_large_ev} holds, where $H = \{ \Nwalkew_i \sep \Nwalkew_i > (1 - 4 \delta) \}$:
    \begin{equation*}
      \sum_{i \in H} \alpha_i^2 > \frac{3\vol{S}}{4} \fstop
    \end{equation*}
    It holds that
    \begin{equation*}
      \sum_{i \in H} (\alpha_i - \beta_i)^2
      \leq \sum_{i} (\alpha_i - \beta_i)^2
      = \normb{\vecto{u}_S - \vecto{u}_T}^2
      = \sum_{x \in S} \degr{x} - \sum_{x \in T} \degr x
      = \vol{S} - \vol{T} = \theta \vol{S} \fstop
    \end{equation*}
    Using the triangle inequality $\normb{a - b} \geq \normb{a} - \normb{b}$ on the subspace spanned by the basis $H$, we have
    \begin{align*}
      \sum_{i \in H} \beta_i^2
      \geq \left[ \sqrt{\sum_{i \in H} \alpha_i^2} - \sqrt{\sum_{i \in H} (\alpha_i - \beta_i)^2} \right]^2
      &> \left[ \sqrt{\frac{3\vol{S}}{4}} - \sqrt{\theta \vol{S}} \right]^2 \\
      &= \frac{3\vol{S}}{4} - \sqrt{\frac{3\theta}{4}} \vol{S} + \theta \vol{S} \\
      &> \frac{11}{20} \vol{S}
    \end{align*}
    Observe that $\beta_1^2 \geq \alpha_1^2 \geq \frac{\vol{S}}{2}$.
    Similar to \cref{eq:spec_alpha_bound}, we have
    \begin{equation*}
      \sum_{\substack{i \geq 2\\ \Nwalkew_i > \tau}}^{n} \left( \sum_{x \in T} \sqrt{\degr{x}}\Nwalkev_i(x) \right)^2
      = \sum_{i \in H} \beta_i^2 - \beta_1^2
      \geq \frac{\vol{S}}{20} \fstop
    \end{equation*}
    Therefore, we obtain that
    \begin{equation*}
      \frac{1}{s} \sum_{x \in T} \normb{\Walk^\ell(x,\cdot) - \statdist}^2
      \geq \frac{\vol{S}}{80 m^2 n^3 s^2} (1 - 4 \delta)^{2 \ell} \fstop
    \end{equation*}
  \end{proof}
  
  \thmspeccorollary*
  \begin{proof}
    Apply \cref{thm:exists_set_weak_vertices}, mark the weak vertex that has been found in $T$ and exchange it for some vertex in $S \minus T$ that has not been marked yet. By \cref{thm:exists_set_weak_vertices}, we can repeat this process until $\vol{T} \geq \theta \vol{S}$.
  \end{proof}
  
    We sketch a proof here. The algorithm is as described in Section~\ref{sec:unknown-size-graph}.
  \section{Proofs from \Cref{sec:lower_bound}}
  
  \subsection{Proof of \cref{thm:lower_bound_implication}}
  \label{proof:thm_lower_bound_implication}
  
  \thmlowerboundimplication*
  \begin{proof}
    We make a simple modification to $\dalg{A}$ to obtain $\dalg{B}$: In the first local computation phase, $\dalgon{B}{v}$ draws a random number $id_v$ uniformly from $\{ 1, \ldots, n^3 \}$ and feeds it into $\dalgon{A}{v}$ as $I(v)$.  Then, $\dalgon{A}{v}$ is executed as normal.  For $u,v \in V$, the probability that $id_u$ and $id_v$ are equal is $1/n^3$.  Applying a union bound, with probability $1 - o(1)$, it holds that $id_u \neq id_v$ for every $u,v \in V$. We then run algorithm $\mathcal{A}$ on this new instance and output the result.
  \end{proof}

\end{document}